\documentclass[fleqn,12pt]{wlscirep}

\usepackage{graphicx}
\usepackage{subfig}
\usepackage{amsmath,mathptmx}
\usepackage{mathrsfs} 
\usepackage{amsfonts}
\usepackage{amssymb}%
\usepackage{amsthm}
\usepackage{amscd}
\usepackage{verbatim}
\usepackage{url}
 \usepackage{dsfont}
\usepackage[all]{xypic} 
\usepackage{endnotes}
\usepackage{enumerate}
\usepackage{mathtools}
\usepackage{comment}
\usepackage{wrapfig}
\usepackage{tcolorbox}
\usepackage{lipsum}
\usepackage{color}

\newcommand{\C}{\mathbb{C}}

\newcommand{\ZZ}{\mathbb{Z}}

\newcommand{\GL}{\operatorname{GL}}

\newcommand{\Mat}{\operatorname{Mat}}

\newcommand{\tn}{\textnormal}

\newcommand{\End}{\textnormal{End}}

\newcommand{\actson}{\curvearrowright}

\theoremstyle{plain}
\newtheorem{theorem}{Theorem}
\newtheorem{corollary}[theorem]{Corollary}
\newtheorem{proposition}[theorem]{Proposition}

\theoremstyle{definition}
\newtheorem{defn}{Definition}
\theoremstyle{definition}

\theoremstyle{definition}

\theoremstyle{definition}
\newtheorem{example}{Example}
\theoremstyle{definition}

\theoremstyle{definition}

\newenvironment{WrapText}[1][r]
  {\wrapfigure{#1}{0.5\textwidth}\tcolorbox}
  {\endtcolorbox\endwrapfigure}

\title{Topological classification of time-asymmetry in unitary quantum processes}

\author[1,$\dagger$]{Jacob Turner}
\author[2,*]{Jacob Biamonte}
\affil[1]{Korteweg-de Vries Institute for Mathematics\\ University of Amsterdam, 1098 XG Amsterdam, Netherlands}
\affil[2]{Skolkovo Institute of Science and Technology\\ 30 Bolshoy Boulevard, Moscow 121205, Russian Federation}

\affil[*]{j.biamonte@skoltech.ru, http://quantum.skoltech.ru}
\affil[$\dagger$]{j.w.turner@uva.nl}

\begin{abstract}
Understanding  which physical processes are symmetric with respect to time inversion is a ubiquitous problem in physics. In quantum physics, effective gauge fields allow emulation of matter under strong magnetic fields, realizing the Harper-Hofstadter, the Haldane models, demonstrating one-way waveguides and topologically protected edge states.  Central to these discoveries is the chirality induced by time-symmetry breaking.  In quantum walk algorithms, recent work has discovered implications time-reversal symmetry breaking has on the transport of quantum states which has enabled a host of new experimental implementations.  We provide a full topological classification of the Hamiltonians of operators breaking time-reversal symmetry in their induced transition probabilities between elements in a preferred site-basis.  We prove that a quantum process is necessarily time-symmetric for any choice of time-independent Hamiltonian precisely when the underlying support graph is bipartite or no Aharonov-Bohm phases 
are present in the gauge field.  We further prove that certain bipartite graphs exhibit transition probability suppression, but not broken time-reversal symmetry. Furthermore, our development of a general framework characterizes gauge potentials on combinatorial graphs. These results and techniques fill an important missing gap in understanding the role this omnipresent effect has in quantum information and computation.
\end{abstract}
\begin{document}

\flushbottom
\maketitle

\thispagestyle{empty}

\section*{Introduction}

Synthesis and emulation of artificial gauge fields for spinless particles has been realized in a number of experimental systems \cite{lim2016electrically}.  These experiments demonstrate that an artificial gauge field can be used to control an observed chirality induced by time-symmetry breaking. This in turn has profound implications for probability transfer control in quantum technologies and quantum algorithms based on quantum walks  \cite{mulken2011continuous,mulken2006coherent}. Indeed, the marriage of these two disciplines will have impacts ranging from applications to transport in topological systems, gauge theories,  topological quantum computing, topological insulators, and the fractional quantum hall effect. The present work provides a classification of any quantum processes into a topological dichotomy, found by a paradigmatic change in how time-symmetry theory is addressed in quantum theory.  Our results were found by developing a powerful framework based on geometric invariant theory.

Quantum walks \cite{kempe2003quantum,venegas2012quantum} are an established tool in quantum physics, quantum computing, and quantum information to study probability transfer. In quantum computing, they have been used to develop quantum search algorithms generalizing Grover's algorithm \cite{wong2015faster,wong2015grover,shenvi2003quantum,krovi2015quantum,childs2004spatial}, as an improvement over classical random walks and Markov processes \cite{moore2002quantum,childs2003exponential,szegedy2004quantum}, and quantum walks represent a universal model of quantum computation \cite{childs2009universal}.

Quantum walks also provide an algorithmic lens for studying quantum transport in physical systems \cite{mulken2011continuous,mulken2006coherent}; for example photosynthetic complexes \cite{mohseni2008environment,rebentrost2009environment}.  Most recently the investigation of quantum transport phenomena has been expanded to so called, `chiral quantum walks' \cite{Z13}, in which time symmetry is broken during unitary evolution ~\cite{Z13,xiang2013transfer,bedkihal2013transfer,manzano2013transfer, Wong15, cameron2013transfer}.

With respect to the probability transport between graph nodes, we fully classify time-reversal symmetry for quantum processes on graphs. While in actuality, a process of this type can describe any unitary process (after a change of basis), we choose to view them as so called `quantum walks' in this paper, leading to a different notion of time-symmetry than usually considered \cite{Z13}. In our notion of time-symmetry, measurement is implicit. It was recently argued in \cite{oreshkov2015operational} that this should be the case as post-selection is the natural analog to preparation when time is reversed.
 
An interesting example of how the notion of time-symmetry classified here is different is the fact that we find the famous Hofstadter model to be time-symmetric as it is a walk on a bipartite graph. This contrasts with the more standard interpretation where this model is not time-symmetric \cite{Hofstadter1976energy} due to eigenfunctions occurring in conjugate pairs sharing the same energy.    

The reason for this is that our main motivation is understanding transition probabilities. In the field of quantum walks, this is the most natural thing to consider as it allows analysis of quantum walks based processes along the same lines as stochastic processes. This is especially important in determining the ``speed'' at which a quantum walk converges to the desired state. 
As a result, we cannot distinguish processes that differ by complex conjugation.

The behavior of the fundamental laws of physics under time-reversal has long
remained central to the foundations of
physics~\cite{wigner59,2013PhRvL.111b0504S} and has a host of applications in condensed
matter theory~\cite{peierls1993,Hofstadter1976energy,sarma2008hall,hasan2010topological,dalibard2011artificial}.

The practical importance of time-reversal symmetry breaking 
is demonstrated by its equivalence to introducing biased probability flow in a quantum system \cite{Z13}. 
Time-symmetry breaking enables directed state transfer
without requiring a biased (or non-local) distribution in the initial states, or coupling to an
environment~\cite{Z13,plenio2008dephasing, 2012PhRvL.108b0602S, Bose07}.  The effect is however subtle:  it is readily shown that time-reversal asymmetry cannot affect the site-to-site transport in some simple cases, such as a Hamiltonian with a support graph representing either a linear chain or a tree.  
On the other side, it is easily demonstrated that the effect is present in a range of physically and practically relevant scenarios.  As such, it is of great interest to understand when such time-reversal symmetry breaking can occur.

In this paper, we do a detailed analysis of a set of fundamental symmetries preserving a preferred on-site basis and transition probabilities in the network topology. The symmetries are also very important in gauge theory, where they coincide with gauge symmetries in materials like graphene.  We recover known effects in gauge theory, with rigorous mathematical proof, such as impossibility of time-symmetry in the presence of magnetic fields and the Aharonov-Bohm effect \cite{aharonov1959significance}. 

Furthermore, our recovery of these facts uses geometric invariant theory and differs from theoretical tools used in gauge theory. As such, we have found a new framework in which these phenomena can be understood. A consequence of our approach is a complete classification of gauge potentials on combinatorial graphs\cite{harrison2011quantum}.  However, our situation differs from the classical scenario by the aforementioned reasons. Indeed, we shall see that time-symmetry can exist for walks even in the presence of a magnetic flux if the walk takes place on a bipartite graph, which is the case for the Hofstadter model.

We establish directly and fully classify the Hamiltonians which enable the breaking of time-reversal 
symmetry in terms of their transition probabilities between elements in a preferred site-basis.


\subsection*{The Probability Transfer Problem}

A continuous time quantum walk~\cite{mulken2011ctqw,PhysRevX.3.041007,childs2009universal,Perseguers2010} is a quantum operator $U(t)=e^{-itH}$ acting on a normalized vector---for quantum search the normalized all-ones vector and for transport, a state in the site-basis. The Hamiltonian $H$ is an adjacency matrix of a weighted bidirected graph: the Hermitian constraint being that $H=H^\dagger$. We call this graph the support of $H$, which we define more formally later.

The quantum evolution of $U(t)$ induces a walk on the graph determined by $H$. In the quantum walks literature, $H$ is typically considered as being symmetric, with real valued (and possibly even negative) edge weights: but it needn't be.  More precisely, an edge $e$ connecting vertices $v$ and $w$ can have conjugate weights with respect to $H$ depending if it is viewed as an edge from $v\to w$ or vice versa. Unless explicitly stated, we will assume that the support of our Hamiltonian is connected and simple.  Immediately one then asks, ``when can such a process induce time-asymmetric evolutions?''  

The question of which processes induce transition probabilities that are symmetric with respect to time must be made more precise: Consider a set $\mathcal{L}\subseteq\End(V)$, where $V$ is a complex Hilbert space with a chosen orthonormal basis $B$. We are given a function $P:B\times B\times\mathcal{L}\to\mathbb{R}_{\ge0}$. Equivalently, $\mathcal{L}\subseteq\Mat_{|B|\times|B|}(\C^{|B|})$, which is to say that $M\in\mathcal{L}$ is matrix from $\C^{|B|}$ to itself expressed in the preferred basis $B$.  Often this function is written as $P_{s\to f}(M)$ where $s$ and $f$ are basis vectors of $V$ and $M\in\mathcal{L}$. The problem is then to classify those operators $M(t)$ such that $P_{s\to f}(M(t))-P_{f\to s}(M(t))=0$ for all basis vectors $f,s\in B$ and all times $t$.

From the point of view of quantum walks,  basis vectors in $B$ are identified with vertices in a graph defined by $M(t)$, on which we consider probabilistic walks. $P_{s\to f}(M(t))$ gives the probability of moving from node $s$ to node $f$ at time $t$. We now look at a two more classical examples that use matrices to define dynamics on graphs in a similar fashion. In each of our examples, $t$ represents a time parameter that allows the walk to evolve.

\vspace{.3cm}
\noindent{\textbf{Classification of Time-Symmetric Stochastic Processes:}}\label{ex:stochastic}
\noindent Motivationally, we examine this same classification for a subset of stochastic processes. Let $S$ be a valid infinitesimal stochastic generator of a continuous-time Markov process, so $s_{ij}\geq 0$ for $i\neq j$, and $\sum_j q_{ij}=0$ for all $i$.  Then the set of matrices of the form $U(t)=e^{tS}$ is a stochastic semi-group. As a caveat, not all stochastic matrices arise this way \cite{davies2010embeddable}, but all such processes considered in this paper will be. This is because every stochastic matrix of this form is invertible, which is not the case in general, although it's inverse may not necessarily be a stochastic process as well. We define the probability function $P_{s\to f}=U(t)_{fs}=\langle f| U|s\rangle$. The \emph{stochastic probability current} is defined by $\hat{s}_{s,f}(U(t)):=U(t)_{fs}-U(t)_{sf}=0$. Then the 
following are equivalent:

\begin{itemize}
 \item $\hat{s}_{s,f}(U(t))=0$ for all basis elements $s$, $f$, and all times $t$.
 \item $S=S^T$, i.e.\ $S$ is a Dirichlet operator. 
\end{itemize}
 
\vspace{.3cm}
Requiring $U(t)$ to be doubly stochastic isn't sufficient for $\hat{s}_{s,f}(U(t))=0$, since there exist non-symmetric doubly-stochastic matrices with positive eigenvalues, which then have a real logarithm. 
Zero stochastic probability current is a strictly stronger condition than the reversibility of a Markov chain at some stationary distribution, implying time symmetry at any initial distribution.

\vspace{.3cm}
\noindent{\textbf{Classification of Amplitude Symmetry:}}\label{ex:quantamp}
\noindent Another related, but less well known example comes from quantum processes. Let $H$ be Hermitian and define $U(t)=e^{-itH}$. The \emph{quantum amplitude current} is defined to be $\hat{q}_{s,f}(U(t)):=U(t)_{fs}-U(t)_{sf}=0$. Then the following are equivalent:
 \begin{itemize}
  \item $\hat{q}_{s,f}(U(t))=0$ for all $s$, $f$, and $t$.
  \item $H=H^T$.
  \item $[H,K]=0$, where $K$ is the anti-unitary operator defined by complex conjugation in the same basis as $H$.
 \end{itemize}

\vspace{.3cm}

These two examples are mathematically similar, and in fact the first example is a special case of the second. The set of stochastic generators $S$ that are also symmetric matrices forms a strict subset of symmetric matrices, so $\hat{s}_{s,f}(U(t))=0$ for  all $s,f$ and $t$ implies that $\hat{q}_{s,f}(U(t))=0$ for all $s,f$, and $t$.

\subsection*{Quantum Probability Current}

The goal of this paper is to study a transition symmetry problem introduced in \cite{Z13}. As in the classification of amplitude symmetry, we let $U(t)=e^{-itH}$ for $H$ Hermitian. We define $P_{s\to f}(U(t))=|U(t)_{sf}|^2=|U(t)_{fs}^\dagger|^2$. If the quantum walk starts with a particle at vertex $s$, then $P_{s\to f}(U(t))$ gives the probability of seeing vertex $f$ after applying $U(t)$ and performing a measurement. The \emph{quantum probability current} is defined as $\gamma_{sf}(U(t)):=|U(t)_{fs}|^2-|U(t)^\dagger_{sf}|^2$. 

\begin{defn}
 A quantum process is \emph{time-symmetric} if and only if $\gamma_{sf}(U(t))=0$ for all $s,f$ and $t$.
\end{defn}

 Both time-symmetric stochastic processes and amplitude symmetric processes are special cases of processes symmetric with respect to the  quantum probability current.  Indeed, we see that for a quantum process with unitary operator $e^{-itH}$, if $H=H^T$, then $H$ must be real and thus trivially time-symmetric. So both of the aforementioned examples arise as subsets of time-symmetric chiral quantum walks \cite{Z13}. However, we shall see that there are other walks which are time-symmetric that do not fall into the above situations will be shown.

We shall give a precise characterization of those Hamiltonians that enable time-symmetry breaking by viewing them as the adjacency matrix of a weighted graph. For the previous examples, time-symmetry was equivalent to the Hamiltonian defining an bidirected weighted graph (which means it can be viewed as undirected). In the general setting, there are further essential conditions the graph satisfy.

\subsection*{The Vanishing of the Quantum Probability Current}\label{sec:probcurr}
An operator $U(t)$ is time-symmetric if and only if there exists a diagonal unitary matrix $\Lambda$ such that either $\Lambda U(t)\Lambda^\dagger=U^T(t)$ or $\Lambda U(t)\Lambda^\dagger=U^\dagger(t)$. Equivalently, we say that an operator is time-symmetric if $U(t)$ can be related to $U(t)^T$ by diagonal unitary matrices and/or complex conjugation in the basis of $B$.

 We call this action the $\Lambda$-action and we will completely determine when a unitary matrix $U$ can be taken to $U^\dagger$ or $U^T$ by this action.  Our classification does not depend on the value of $t$ (unless $t=0$), and so we simply consider unitary matrices instead of unitary operators throughout the rest of this paper.

 \begin{defn}
 The \emph{support} of an $n\times n$ matrix $M$ is a digraph with $n$ vertices and an edge from vertex $i$ to vertex $j$ if $M_{ij}\ne0$. We denote it by $supp(M)$.
\end{defn}

Note that for a Hermitian matrix $H$, the support may be viewed as undirected as $H_{ij}\ne0$ if and only if $H_{ji}=\overline{H_{ij}}\ne0$. We can think of $H$ as a weighted adjacency matrix. Then we have the weight of traveling in one direction along an edge is the complex conjugate of the weight of traveling the opposite direction.

Our classification will be in terms of the Hamiltonian of an operator. This is because $\Lambda e^{-itH}\Lambda^\dagger=e^{-it\Lambda H\Lambda^\dagger}$ and $\overline{e^{-itH}}=e^{it\overline{H}}$. So each of the symmetries we consider induces a natural action on the Hamiltonians defining the operator.

\section*{Results}

The presence of Aharonov-Bohm phases are a known obstacle to time-symmetry in gauge theory. Much work has been done to understand gauge theory on metric graphs, and recently this was simplified further to understanding gauge potentials on combinatorial graphs \cite{harrison2011quantum}. In this setting, one considers the graph to be support of the Hamiltonian $H$. The $\Lambda$-action on $H$ induces an action on the model which corresponds precisely to the gauge transformations in this setting. We call this model the Harrison, Keating, Robbins (HKR) model after the authors; see to side box.

\begin{WrapText}
 \framebox{\textbf{Classification of HKR model up to Gauging.}}
 
 \begin{itemize}
  \item Given a Hamiltonian, the HKR model defines a tight binding model on $\tn{supp}(H)$ whose configuration space is the set of vertices.
  \item  If we write $H=\{r_{jk}e^{i\theta_{jk}}\}$, then let $\Omega=\{\theta_{jk}\}$, with $\theta_{jk}=0$ if $r_{jk}=0$.
  \item This is a skew-symmetric matrix which is defined as a gauge potential on $\tn{supp}(H)$ \cite{harrison2011quantum}.
  \item Theorem \ref{thm:invseparate} states that two Hamiltonians, $H_1$ and $H_2$, define the same instance of the HKR model up to gauge transformations if and only if $w_c(H_1)=w_c(H_2)$ for all cycles $c\in\tn{supp}(H_1)$. 
  \item This implies the result that a Hamiltonian $H$ defines a trivial gauge potential (up to gauge transformations) if and only if all the Aharonov-Bohm phases are zero.
 \end{itemize}

\end{WrapText}

Understanding when $H$ is equivalent to a real matrix under the $\Lambda$-action is the same as knowing when the gauge potential induced by $H$ is equivalent to a trivial one. This question was solved in \cite{harrison2011quantum}, however the answer is also implied by our more general results. We are also interested in knowing when $H$ and $\overline{H}$ can be related by the $\Lambda$-action for the following reason.

As previously mentioned, fixing a basis $B$ defines a complex conjugation involution on linear maps, $M \mapsto \overline{M}$. 
It is always the case that $P_{s\to f}(\overline{U(t)})=P_{s\to f}(U(t))$. But if $H$ and $\overline{H}$ can be related by the $\Lambda$-action, we have that
\begin{eqnarray}
P_{s\to f}(\overline{U(t)})&=&|f^\dagger \overline{e^{-itH}} s|^2\\ &=&|f^\dagger e^{it \overline{H}} s|^2=|f^\dagger e^{itH}s|^2.\nonumber
\end{eqnarray}

So if $H$ is conjugate via the $\Lambda$-action to its transpose, then $H$ is time-symmetric combining the $\Lambda$-action with the action of 
complex conjugation in our chosen basis.

\begin{defn}
 Let $H$ be Hermitian and $c=i_1\to i_2\to\cdots\to i_k$ be a path in $supp(H)$. Then we define the weight of $p$ to be $w_c(H)=H_{i_1i_2}\cdots H_{i_{k-1}i_k}$. We define the \emph{Aharonov-Bohm phase} of a cycle $c$ to be the complex angle of $w_c(H)$.
\end{defn}

The Aharonov-Bohm phase defined above coincides with the classical notion of the Aharonov-Bohm phase for the HKR model. As we will see, these phases completely determine the gauge potentials up to gauge transformation. However, they do not completely determine when two Hamiltonians are equivalent under the $\Lambda$-action. For that, the values $w_c(H)$ are necessary. We see that the polynomials $w_c(H)$ are polynomial invariants of the Hamiltonians under the gauge transformations and in fact they form a complete set of invariants.

\begin{theorem}\label{thm:invseparate}
 Two Hamiltonians $H_1$ and $H_2$ are related by the $\Lambda$-action if and only if $w_c(H_1)=w_c(H_2)$ for all $c$. Furthermore, two gauge potentials can be related by a gauge transformation if and only if their Aharonov-Bohm phases coincide.
\end{theorem}
\begin{theorem}\label{thm:realinvs}
 A Hermitian matrix $H$ takes real values on its invariants if and only if it is conjugate to a real matrix via the $\Lambda$-action. Furthermore, $H$ is conjugate to $\overline{H}$ if and only if $H$ is also conjugate to a real matrix.
\end{theorem}
\begin{corollary}\label{cor:realinvs}
 If a Hermitian matrix takes real values on all of its invariants, it is time-symmetric.
\end{corollary}

Theorem \ref{thm:invseparate} completely classifies gauge potentials on combinatorial graphs. Theorem \ref{thm:realinvs} tells us that a gauge potential $\Omega$ is gauge equivalent to $\Omega^T$ if and only if it is equivalent to a trivial gauge potential. Corollary \ref{cor:realinvs} confirms what one might expect from a gauge theoretic perspective; that time-symmetry is closely related to having no non-trivial Aharonov-Bohm phases.  

It is surprising, therefore, that having a trivial gauge potential is not the only way to have time-symmetry. Suppose that we have a Hamiltonian $H$ such that there is a $\Lambda$ with $\Lambda H\Lambda^\dagger=-H$. Then $\Lambda U(t)\Lambda^\dagger=U(-t)$, allowing us to reverse time with a gauge transformation.

\begin{theorem}\label{thm:conjtoneg}
 A Hamiltonian $H$ with zero diagonal is conjugate to its negative under the $\Lambda$-action if and only if its support is bipartite.
\end{theorem}

Note that this necessarily requires that $H$ have zeros on the diagonal. However, one can trivially see that for $\alpha\in\mathbb{R}$, \begin{eqnarray}
P_{s\to f}(e^{-it(H+\alpha I)})&=&P_{s\to f}(e^{-i\alpha tI}e^{-itH})\\
&=&P_{s\to f}(e^{-itH}).\nonumber
\end{eqnarray}
 So if $H$ is time-symmetric, $H+\alpha I$ is as well. In particular, the $supp(H)$ can be bipartite with self-loops if the every loop has the same weight. However, one cannot add arbitrary diagonal matrices to $H$ and preserve time symmetry.
 
 \begin{example}\label{ex:disorder}
The following Hamiltonian is not time-symmetric however because the diagonal contains distinct entries, although if the diagonal is made to be all zeros, then the support is bipartite.
\begin{equation*}
 \begin{pmatrix}
  1&1&0&i\\
  1&2&1&0\\
  0&1&3&1\\
  -i&0&1&4
 \end{pmatrix}
\end{equation*}
\end{example}
\vspace{.3cm}
\subsection*{Disorder and Phase Independence}

Using the invariant techniques we developed, we can quickly answer two other interesting questions. Suppose disorder is added to a system by changing the energies of nodes, i.e. having self loops of different weights. This corresponds to adding a real diagonal matrix to the Hamiltonian. As Example \ref{ex:disorder} shows, adding disorder can break time symmetry. The following proposition gives a necessary condition for this to not break time symmetry.
\begin{proposition}\label{prop:adddiag}
If $H$ is conjugate to a real matrix via the $\Lambda$-action, then for any real diagonal matrix $D$, $H+D$ is time-symmetric.
\end{proposition}

If we are given a Hamiltonian $H=(h_{ij}e^{i\alpha_{ij}})$ with $h_{ij},\alpha_{ij}\in\mathbb{R}_{\ge0}$, the second interesting question is how the transition probabilities are affected by the choice of phases $\alpha_{ij}$. It was shown in \cite{Z13} that if the underlying graph of $H$ is a tree, then surprisingly, the transition probabilities are not affected by the choice of $\alpha_{ij}$.

\begin{proposition}\label{prop:alphaind}
 The transition probabilities of a walk is independent of the choice of $\alpha_{ij}$ if and only if the (undirected) support of $H$ is a tree (with possible self-loops). 
\end{proposition}

\section*{Discussion}

Our theory provides more than a classification of quantum processes and should have applications in several related fields.  For example, our work can be related to the theory of tomography where it can serve as a theory for classification.

The doubly stochastic matrices arise by taking a unitary matrix $U=\{u_{ij}\}$ to the matrix $\{|u_{ij}(t)|^2\}$. Understanding the image of this map inside the Birkhoff polytope, called the set of unistochastic matrices, has been an important problem in tomography. Of note is the problem of reconstructing the Cabibbo-Kobayashi-Masakawa matrix describing quark decay \cite{cabibbo1963unitary,kobayashi1973cp}. Experimentally only the transition probabilities can be determined \cite{jarlskog1985commutator,dita2005global}. This matrix was important in demonstrating another symmetry violation, namely CP-violation, and lead to a Nobel Prize in Physics. Work on reconstructing the Maki-Nakagawa-Sakata matrix in neutrino physics is still an unfinished \cite{maki1962remarks}. The unistochastic matrices also are important in scattering theory \cite{mennessier1974some} and quantum information \cite{werner2001all}.

An important aspect of the map taking the unitary group to the set of unistochastic matrices are natural symmetries on the fibers. Every fiber has set of fundamental symmetries sometimes called Haagerup equivalence, although these do not generally account for all of the symmetries. In this subject, more focus has been placed on understanding the exotic symmetries. However, the Haagerup equivalence symmetries are omnipresent in these problems. Furthermore, they correspond to the symmetries classified with the theory in this paper: the $\Lambda$-action and complex conjugation.

\section*{Methods}

We now develop the necessary techniques and provide proofs of our results. If $U(t)$ is time-symmetric, we must have $P_{s\to f}(U(t))=P_{s\to f}(U^T(t))=P_{s\to f}(U^\dagger(t))$ for all $s,f\in B$ and times $t$. We approach this problem by understanding the symmetries (and what they leave invariant) that preserve the preferred basis and take unitary operators to unitary operators.
 
Since $U(t)$, $U^T(t)$, and $U^\dagger(t)$ are all unitary operators, if they are related by a symmetry, it must be a completely positive map. Furthermore, since we demand that this map holds for all time $t$, we see that there must exist a unitary matrix $V$ such that either $VU(t)V^\dagger=U^T(t)$ or $VU(T)V^\dagger=U^\dagger(t)$. This is because the only linear maps commuting with matrix multiplication is conjugation and scaling. However, we have fixed the diagonal of our Hamiltonian to be zero so we cannot scale $U(t)$ non-trivially.

 Lastly, since we are working in a preferred basis $B$, $V$ must be a combination of a permutation of the basis elements of $B$ and a matrix that has this basis as eigenvectors. However, since we label our basis vectors, conjugation by a permutation matrix is simply a relabeling and thus does nothing with respect to our classification. This means that $V$ can be taken to be a unitary matrix that is diagonal in the basis $B$.

Therefore, we have that $U(t)$ is time-symmetric if and only if there exists a diagonal unitary matrix $\Lambda$ such that either $\Lambda U(t)\Lambda^\dagger=U^T(t)$ or $\Lambda U(t)\Lambda^\dagger=U^\dagger(t)$. Equivalently, we say that an operator is time-symmetric if $U(t)$ can be related to $U(t)^T$ by diagonal unitary matrices and/or complex conjugation in the basis of $B$.

In order to prove Theorems \ref{thm:invseparate} and \ref{thm:realinvs}, we must first establish a few lemmas and rely on techniques from geometric invariant theory.

The $\Lambda$-action is the same as conjugation by $U(1)^n$, if $H$ is $n\times n$. This is the group of gauge transformations for the HKR model. We consider the polynomial invariants of this action which are generated by the weights of cycles as we defined above (which are allowed to repeat edges or vertices) in the support of $H$ of length $\le n$.

 More formally, let $\Mat_n(\C)$ denote the set of $n\times n$ matrices in basis $B$.  Then the polynomials $H_{i_1i_2}\cdots H_{i_ki_1}$, where $1\le k\le n$ generate the invariant ring $\C[\Mat_n(\C)]^{U(1)^n}$, where $\C[\Mat_n(\C)]^{U(1)^n}$ is the ring of polynomials that are constant on the orbits of $U(1)^n$. We emphasize that invariants of the form $H_{ij}H_{ji}=|H_{ij}|^2$ and $H_{ii}$ are included. However, these invariants are necessarily real and so we call these the \emph{trivial invariants}.

We claim that two Hermitian matrices are in the same $U(1)^n$ orbit if and only if they have the same values on all invariants. It is easy to see that if two Hamiltonians differ in their invariants, they cannot be in the same orbit as these functions are constant on orbits. Furthermore, a function that is constant on a set is also constant on the closure of that set, so orbits whose closures intersect cannot be distinguished via invariants. We prove that this does not happen. 

We mention that since we are considering polynomial invariants, the correct topology to use is, \emph{a priori}, the Zariski topology. However, in this setting, it is well known that the closure of an orbit in the Zariski and Euclidean topologies will coincide. 

We wish to show that the orbit closures of two Hermitian matrices do not intersect. We first consider the orbit closures of Hermitian matrices under the action of the closely related group $\GL(1)^n$. This group action has exactly the same invariant polynomials as $U(1)^n$ as it is well known that the unitary group is Zariski dense inside the general linear group of the same dimension. This is sometimes called Weyl's trick. See for example \cite{kraft2000classical}.

\begin{theorem}[Hilbert-Mumford Criterion \cite{kempf1978instability}]\label{thm:hmcrit}
 If $G$ is a product of general linear groups acting on a complex vector space $V$, then  if $\overline{G.v}-G.v\ne\emptyset$, there is some $y\in\overline{G.v}-G.v\ne\emptyset$ and a 1-parameter subgroup $\lambda:\C^\times\to G$ such that $\lim_{t\to 0}{\lambda(t).v}=y$.
\end{theorem}

For the action of of $\GL(1)^n\actson\Mat_n(\C)$ by conjugation, the only 1-parameter subgroups in $\GL(1)^n$ are diagonal matrices of the form $\lambda(t)_{ii}=t^{\alpha_i}$, $\alpha_i\in\ZZ$ (cf. \cite{kraft2000classical}). If the orbit of $M$ is closed, then there is no 1-parameter subgroup taking $M$ outside of its orbit by Theorem \ref{thm:hmcrit}. If $\lambda(t)$ is diagonal with $\lambda(t)_{ii}=t^{\alpha_i}$, then $\lambda(t)M\lambda(t)^{-1}=\{t^{\alpha_i-\alpha_j}m_{ij}\}$. If $\lim_{t\to0}{\lambda(t)M\lambda(t)^{-1}}$ exists, then no negative power of $t$ appears in $\lambda(t)M\lambda(t)^{-1}$. Furthermore, the limit sends some of the entries of $M$ to zero and leaves the rest unchanged.
 
\begin{proposition}\label{cor:lambdaclosed}
 Hermitian matrices have closed orbits under the action of $\GL(1)^n$.
\end{proposition}
\begin{proof}
Let $\lambda(t)$ be a 1-parameter subgroup such that $\lim_{t\to 0}{\lambda(t)M\lambda(t)^{-1}}$ exists, $M$ Hermitian. Then suppose $t^{\alpha_i-\alpha_j}m_{ij}$, $\alpha_i>\alpha_j$, is an non-zero entry of $\lambda(t)M\lambda(t)^{-1}$. Then so is $t^{\alpha_j-\alpha_i}m_{ji}$. Thus it's easy to see that as $t\to 0$, $t^{\alpha_j-\alpha_i}m_{ji}$ goes to infinity. So it must be that $\lambda(t)M\lambda(t)^{-1}=M$, implying it has a closed orbit.
\end{proof}

The following proposition tells us that by restricting our view to matrices that have closed orbits, all the information we need is given by the invariant polynomials.

\begin{proposition}[\cite{MR1304906}]\label{prop:closedorbits}
 For a product of general linear groups acting on a complex vector space $V$, if $M_1$ and $M_2$ have closed orbits, they are in the same orbit if and only if they agree on all invariants.
\end{proposition}

\noindent\textbf{{Proof of Theorem \ref{thm:invseparate}}}
\begin{proof}
The first assertion follows from Propositions \ref{cor:lambdaclosed} and \ref{prop:closedorbits} along with the observation that two Hermitian matrices are similar via a change of basis if and only if they are similar via a unitary change of basis.

For the second assertion, we again note that the norms of the entries of a Hamiltonian are left unaffected by the $\Lambda$-action. Given gauge potentials $\Omega=\{\theta_{jk}\}$ and $\Omega'=\{\theta'_{jk}\}$ with $j,k\in[n]$, we consider the Hamiltonians $H=\{e^{i\theta_{jk}}\}$ and $H'=\{e^{i\theta'_{jk}}\}$. 

For a cycle $c:i_1\to i_2\to\cdots\to i_k\to i_1$ in the complete graph on $n$ vertices $K_n$, let $\omega_c:=\theta_{k1}+\sum_{j=1}^k{\theta_{j,j+1}}$ be the associated Aharanov-Bohm phase for the gauge potential $\Omega$. We define $\omega'_c$ similarly. Then $H$ and $H'$ can be related by the $\Lambda$-action if and only if $w_c(H)=e^{i\omega_c}=e^{i\omega'_c}=w_c(H')$ for all cycles $c$ in $K_n$. This implies that the two gauge potentials are equivalent if and only if $\omega_c=\omega'_c$ (modulo $2\pi$) for all cycles $c$ in $K_n$. That is to say, they agree if and only if their Aharanov-Bohm phases coincide.
\end{proof}
\noindent{\textbf{Proof of Theorem \ref{thm:realinvs}}}
\begin{proof}
 If $H$ is conjugate to a real matrix or its support is a tree, this it is clear that all of its invariants are real. Now suppose $H$ takes real values on all of its invariants. Note that $|H|$ takes the same value on all of these invariants. Since $H$ and $|H|$ are both Hermitian, their orbits are closed and thus there is a $\Lambda\in U(1)^n$ such that $|H|=\Lambda H\Lambda^{\dagger}$ by Theorem \ref{thm:invseparate}.
 
Now let $w_p(H)$ be the invariant defined by looking at the weight of the cycle $p$ in the support of $H$ with edge weights induced by $H$. Then it is clear that $w_p(H)=\overline{w_p(\overline{H})}$. So we see that $H$ and $\overline{H}$ are conjugate if and only if $w_p(H)=\overline{w_p(H)}$ by Theorem \ref{thm:invseparate}. But this happens if and only if all invariants of $H$ are real and by the first assertion of the theorem, is equivalent to the fact that $H$ is conjugate to a real matrix.
\end{proof}

\noindent{\textbf{Proof of Theorem \ref{thm:conjtoneg}}}
\begin{proof}
 Only the `if' direction was proven in \cite{Z13}; we include a proof here for completeness. Suppose that the $supp(H)$ is bipartite. If $H$ is $n\times n$, let the vertices be labeled by $1,\dots, n$. Let $A=\{a_1,\dots,a_k\}\subset[n]$ denote one of the independent sets of $supp(H)$. Let $B$ denote the other. Then let $\Lambda_{ii}=-1$ for $i\in A$ and $1$ otherwise. Then $(\Lambda H\Lambda^\dagger)_{ij}=\Lambda_{ii}H_{ij}\Lambda_{jj}^{-1}$. We know that$H_{ij}$ is 0 unless exactly one of either $i$ or $j$ is in $A$ as its support is bipartite. Then $\Lambda_{ii}H_{ij}\Lambda_{jj}^{-1}=-H_{ij}$. So $\Lambda H\Lambda^\dagger=-H$.
 
 For the converse, suppose that $H$ is conjugate to $-H$, i.e. $\exists\Lambda$ such that $\Lambda H\Lambda^\dagger=-H$. We know from Theorem \ref{thm:invseparate} that this implies $w_c(H)=w_c(-H)$ for all cycles $c$. But $w_c(H)$ is a homogeneous monomial of degree $|c|$, the length of $c$. So $w_c(-H)=(-1)^{|c|}w_c(H)$, implying that $w_c(H)=0$ if $|c|$ is odd. Since $H$ can only have even cycles in its support, it is bipartite. 
 \end{proof}
 
 \noindent{\textbf{Proof of Proposition \ref{prop:adddiag}}}
\begin{proof}
If there is a $\Lambda\in U(1)^n$ such that $\Lambda H\Lambda^\dagger$ is real, then $\Lambda(H+D)\Lambda^\dagger=\Lambda H\Lambda^\dagger+D$ is also real and thus time-symmetric.
\end{proof}

\noindent{\textbf{Proof of Proposition \ref{prop:alphaind}}}
\begin{proof}
First of all, if the support of $H$ is a tree, the only invariants are trivial invariants which specifies the norm of each entry in $H$.  This means that $H$ can always be made real no matter the starting choice of $\alpha_{ij}$. Now suppose $H$ is not a tree and $e^{-itH}$ has a non-zero diagonal entry. Then there is a non-trivial cycle invariant $w(c)$ for a cycle $c:i_1\to i_2\to\cdots i_k\to i_1$, $k>2$. By changing the phase of $i_1$, we change the phase of $w(c)$ and thus we get a continuous family of quantum walks that cannot be related by the group generated by $K$ and the $\Lambda$-action. 
\end{proof}
\section*{Acknowledgements}

The research leading to these results has received funding from the European Research Council under the European Union's Seventh Framework Programme (FP7/2007-2013) / ERC grant agreement No 339109.


\section*{Additional information}

The authors declare no competing financial interests. 

\bibliography{chiral} 

\end{document}